\documentclass{article}
\usepackage{spconf}
\usepackage{amsfonts}
\usepackage{amsmath}
\usepackage{graphicx}
\usepackage{url}
\usepackage{eepic, epsfig, amsmath, amssymb, latexsym, setspace, subfig}
\usepackage{rotating}
\usepackage{amsfonts}
\usepackage{amsmath}
\usepackage{graphicx}
\usepackage{nicefrac}
\usepackage{lipsum,multicol}

\newtheorem{theorem}{Theorem}[section]

\newtheorem{lemma}[theorem]{Lemma}

\numberwithin{equation}{section}

\newcommand{\qed}{\rule{7pt}{7pt}}
\newenvironment{proof}{\noindent{\bf Proof}\hspace*{1em}}{\hfill\qed\vspace{0.125in}}

\newcommand{\x}{\mathbf{x}}
\newcommand{\y}{\mathbf{y}}

\newcommand{\w}{\mathbf{w}}
\newcommand{\e}{\mathbf{e}}

\title{TOEPLITZ MATRIX BASED SPARSE ERROR CORRECTION IN SYSTEM IDENTIFICATION:
OUTLIERS AND RANDOM NOISES}

\name{Weiyu Xu, Er-Wei Bai and Myung Cho} 
\address{Department of ECE, University of Iowa}

\begin{document}
%

\maketitle

\begin{abstract}
In this paper, we consider robust system identification under sparse
outliers and random noises. In our problem, system parameters are observed through a
Toeplitz matrix. All observations are subject to random noises and a few are
corrupted with outliers. We reduce this problem of system identification to a sparse
error correcting problem using a Toeplitz structured real-numbered coding matrix.
We prove the performance guarantee of Toeplitz structured matrix in sparse error
correction. Thresholds on the percentage of correctable errors for Toeplitz structured
matrices are also established. When both outliers and observation noise are
present, we have shown that the estimation error goes to $0$ asymptotically as long as
the probability density function for observation noise is not ``vanishing'' around $0$.
\end{abstract}



\section{Introduction}
In system identification, an unknown system state $\x \in R^m$ is often observed through a Toeplitz matrix $H \in R^{n \times m}$ ($n \geq m$), namely
\begin{equation*}
\y=H\x,
\end{equation*}
where $\y=(y_1,y_2,...,y_n)^T$ is the system output and
the Toeplitz matrix $H$ is equal to
\begin{equation}
\label{Top}
          \begin{bmatrix}
                h_{-m+2} & h_{-m+3} & \ldots & h_{1}\\
                h_{-m+3} & h_{-m+4} & \ldots & h_{2}\\
                 \vdots& \ldots & \ldots & \vdots  \\
                 \vdots& \ldots & \ldots & \vdots  \\
                 \vdots& \ldots & \ldots & \vdots  \\
                                h_{-m+n+1} & \ldots & \ldots & h_{n}
\end{bmatrix},
\end{equation}
with $h_i$, $-m+2\leq i \leq n$, being the system input assumed to be
an i.i.d. $N(0,1)$ Gaussian random sequence.

If there is no interference or noise in the observation $\y$, one can then simply
recover $\x$ from a matrix inversion. However, in applications, all
observations $\y$ are corrupted by noises and a few elements can be exposed to
large-magnitude gross errors or outliers. Such outliers can happen with the failure of measurement devices, measurement
communication errors and the interference of adversary parties.
 Mathematically, when both additive observation noise and outliers are present,
the observation $\y$ can be written as
\begin{equation}
\y=H\x+\e+\w,
\label{eq:model}
\end{equation}
where $\e$ is a sparse outlier vector with $k \ll n$ nonzero elements, and
$\w$ is a measurement noise vector with each element being i.i.d. random variables. We further assume $m$ is fixed, which is often the case in system identifications \cite{Ljung}.

If only random measurement errors are present, the least-square solutions generally
provide an asymptotically good estimate. However, the least-square estimate breaks down
in the presence of outliers.
Thus, it is necessary to protect the estimates from both random noise and outliers.
Research along this direction has attracted a significant amount of attention,
for example, \cite{Bai, cook, Neter, Rousseeuw, Ljung, Stoica}.
In particular for reducing the effects of outliers,
the least absolute deviation estimate ($\ell_1$ minimization)
was proposed and studied \cite{SIC,CandesErrorCorrection,CT1, DMTSTOC, CDC2011}. Instead of searching for all the $\binom{n}{k}$ possibilities for the locations of outliers,  \cite{SIC, CandesErrorCorrection, CT1} proposed to minimize the least absolute deviation:
\begin{eqnarray}
\min_{\x}  &&\|\y-H\x\|_{1}.
\label{eq:errorcorrection}
\end{eqnarray}
Under the assumption that the error $\e+\w$ is an i.i.d.
random sequence with a common density which has median zero and is continuous
and positive in the neighborhood of zero, the difference between the unknown $\x$
and its estimate is asymptotically Gaussian of zero mean \cite{SIC}. The problem is that
the assumption of a common density on the outliers is seldom satisfied in reality.
Also, median zero on $\e+\w$ is restrictive.

In \cite{CandesErrorCorrection, CT1, DMTSTOC, RudelsonVershynin, StojnicThresholds,XuHassibi,CDC2011}, each element
of $H$ (or the nonsingular $(n-m) \times n$ matrix $A$ such that $AH=0$) is assumed to be i.i.d. random variables following a certain
distribution, for example, Gaussian distribution or Bernoulli distribution.
These types of matrices have been shown to obey certain conditions such as restricted isometry conditions \cite{CandesErrorCorrection} so that (\ref{eq:errorcorrection}) can correctly recover $\x$ when there are only outliers present; and can recover $\x$ approximately when both outliers and
measurement noise exist. However, in the system identification problem, $H$ has
a natural Toeplitz structure and the elements of $H$ are correlated. The natural
question is whether (\ref{eq:errorcorrection}) also provides performance
guarantee for recovering $\x$ with a Toeplitz matrix. We provide a positive
answer in this paper.

Though the elements of Toeplitz matrices are correlated, we have shown that Toeplitz structured matrices also enable the successful recovery of $\x$ by using (\ref{eq:errorcorrection}). The main contribution of this paper is the establishment of the performance guarantee of Toeplitz structured matrices in sparse error correction. In particular, we calculated the thresholds on the sparsity $k$ such that an error vector with no more than $k$ nonzero elements can be recovered using (\ref{eq:errorcorrection}).  When both outliers and observation noise are
present, we have shown that the estimation error goes to $0$ asymptotically as long as
the probability density function for observation noise is not ``vanishing'' around $0$.

There is a well known duality between compressed sensing \cite{Neighborlypolytope, DonohoTanner} and sparse error detection \cite{CandesErrorCorrection, CT1}: the null space of sensing matrices in compressed sensing corresponds to the tall matrix $H$ in sparse error corrections. Toeplitz and circulant matrices have been studied in compressed sensing in several papers \cite{BajwaToeplitz}\cite{randomconvolution}\cite{Rauhut}. In these papers, it has been shown that Toeplitz matrices are good for recovering sparse vectors from undersampled measurements. In contrast, in our model of sparse error correction, the signal itself is \emph{not} sparse and the linear system involved is overdetermined rather underdetermined. Also, the null space of a Toeplitz matrix does not necessarily
correspond to another Toeplitz matrix; so the problem studied in this paper is essentially different from those studied in \cite{BajwaToeplitz}\cite{randomconvolution}\cite{Rauhut}.

The rest of this paper is organized as follows.  In Section \ref{sec:strongthreshold}, we derive performance bounds on the number of outliers we can correct when only outliers are present. In Section \ref{sec:noiseerrortogether}, we derive the estimation of system parameters when both gross errors and observation noises are present. In Section \ref{sec:numerical}, we provide the numerical results and conclude our paper by discussing extensions and future directions.

\section{With Only Outliers}
\label{sec:strongthreshold}
We establish one main result regarding the threshold of successful recovery of $\ell_1$-minimization using Toeplitz matrix.
\begin{theorem}
\label{thm:strongmain}
Let $H$ be an $n \times m$ Toeplitz matrix as in (\ref{Top}), where $m$ is a fixed positive integer and $h_i$, $-m+2 \leq i \leq n$ are i.i.d. $N(0,1)$ Gaussian random variables. Suppose that $\y=H\x+\e$, where $\e$ is a sparse vector with no more than $k$ nonzero elements. Then there exists a constant $c_1 >0$ and a constant $\beta>0$ such that, with probability $1-e^{-c_1n}$ as $n \rightarrow \infty$, the $n \times m$ Toeplitz matrix
$H$  has the following property: for every $\x \in R^m$ and every error $\e$ with its support $K$ satisfying $|K|=k \leq \beta n$, $\x$ is the unique solution to (\ref{eq:errorcorrection}). Here the constant $0<\beta<1$ can be taken as any number such that for some constant $\mu>0$ and $0<\delta<1$,
$\beta \log(1/\beta)+(1-\beta) \log(\frac{1}{1-\beta})+m\beta [\log(2)+\frac{m\mu^2}{2}+ \log(\Phi(\mu \sqrt{m}))]+(\frac{1}{2m-1}-\beta) [\log(2)+\frac{1}{2}\mu^2(1-\delta)^2+\log(1-\Phi(\mu(1-\delta)))]<0$, where $\Phi(t)=\frac{1}{\sqrt{2\pi}}\int_{-\infty}^{t}{e^{-\frac{x^2}{2}}\,dx}$ is the cumulative distribution function for the standard Gaussian random variable.
\end{theorem}
\textbf{Remark}: The derived correctable fraction of errors $\beta$ depends on the system dimension $m$.
In the rest of this section, we outline the strategy to prove Theorem \ref{thm:strongmain}. Our derivation is based on checking the following now-well-known theorem for $\ell_1$ minimization (see \cite{Yin}, for example).
\begin{theorem}
 (\ref{eq:errorcorrection}) can recover the correct state $\x$ whenever $\|\e\|_0 \leq k$, if and only if for every vector $z \in R^{m}\neq 0$,  $\|(Hz)_K\|_{1} < \|(Hz)_{\overline{K}}\|_{1}$ for every subset $K \subseteq \{1,2,...,n\}$ with cardinality $|K|=k$, where $\overline{K}=\{1,2,...,n\}\setminus K$.
\label{thm:balanced}
\end{theorem}
The difficulty of checking this condition is that the elements of $H$ are not independent random variables and that the condition must hold for every vector in the subspace generated by $H$. We adopt the following strategy of discretizing the subspace generated by $H$,see \cite{DMTSTOC,StojnicXuHassibi,Ledoux01}.
%
It is obvious that we only need to consider $Hz$ for $z \in R^{m}$ with $\|z\|_2=1$. We then pick a finite set $V=\{v_1, ..., v_{N}\}$ called $\gamma$-net on $\{z|\|z\|_2=1\}$ for a constant $\gamma>0$: in a $\gamma$-net, for every point $z$ from $\{z|\|z\|_2=1\}$, there is a $v_l \in V$ such that $\|z-v_l\|_2 \leq \gamma$. We subsequently establish the property in Theorem \ref{thm:balanced} for all the points in $\gamma$-net $V$ before extending the results to every point $Hz$, where $\|z\|_2=1$.

Following this strategy, we establish Lemmas \ref{cor:s1}, \ref{lemma:blconcentration} and \ref{lemma:slp}. Lemma \ref{lemma:slp} then directly implies Theorem \ref{thm:strongmain}. Most proofs are listed in \cite{Arxiv} for the sake of space. We first show the concentration of measure phenomenon for $Hz$, where $z \in R^m$ is a single vector with $\|z\|_2=1$.
\begin{lemma}\label{cor:s1}
Let $\|z\|_2=1$. For any $\epsilon >0$, there exists a constant $c_2>0$ such that
when $n$ is large enough, with probability $1-2e^{-c_2\frac{n^2}{(n+m-1)m}}$, it holds
that $(1-\epsilon)S \leq \|Hz\|_1\leq (1+\epsilon)S$, where $S=nE\{|X|\}$ and $X$ is a random variable following the Gaussian distribution $N(0,1)$.
\end{lemma}
\begin{lemma}
Let $\|z\|_2=1$ and $0<\delta<1$ be a constant. Then there exists a threshold $\beta \in (0,1)$ and a constant $c_3>0$ (depending on $m$ and $\beta$), such that, with a probability $1-e^{-c_3n}$, for all subsets $K\subseteq \{1,2,...,n\}$ with cardinality $\frac{|K|}{n}\leq \beta$,
\begin{equation*}
\|(Hz)_K\|_1 \leq \frac{1-\delta}{2-\delta} \|Hz\|_1.
\end{equation*}
\label{lemma:blconcentration}
\end{lemma}
By a union bound on the size of $\gamma$-net, Lemma \ref{cor:s1} and \ref{lemma:blconcentration} indicate that with overwhelming
probability the recovery condition in Theorem \ref{thm:balanced} holds for the discrete points on $\gamma$-net.
The following lemma extends the result to $\{z|\|z\|_2=1\}$.
\begin{lemma}\label{lemma:slp}
There exist a constant $c_4>0$ such that when $n$ is large enough, with probability $1-e^{-c_4n}$, the Toeplitz matrix $H$ has the following property: for every $z \in R^m$ and every subset $K
\subseteq \{1,...,n\}$ with $|K| \leq \beta n$, $\sum \limits_ {i \in
\overline{K}} |(Hz)_i| - \sum \limits_{i \in K} |(Hz)_i| \geq \delta' S
$, where $\delta'>0$ is a constant.
\end{lemma}
\begin{proof}
For any given $\gamma>0$, there exists a $ \gamma$-net $V=\{v_1, ..., v_{N}\}$ of
cardinality less than $(1+\frac{2}{\gamma})^m$\cite{Ledoux01}. Since each row of $H$ has $m$ i.i.d $N(0,1)$ entries, elements of $Hv_j$, $1\leq j\leq N $, are
(not independent) $N(0,1)$ entries. Applying a union bound on the size of $\gamma$-net, Lemmas \ref{lemma:blconcentration}
and \ref{cor:s1} imply that for every $v_j \in V$, for some $\delta>0$ and for any constant
$\epsilon>0$, with probability $1-2e^{-cn}$ for some $c>0$,
\begin{eqnarray*}\label{prop:srho}
&&\|(Hv_j)_K\|_1 \leq
\frac{(1-\delta)(1+\epsilon)}{2-\delta}S\\
&&(1-\epsilon)S \leq \|Hv_j\|_1\leq (1+\epsilon)S
\end{eqnarray*}
hold simultaneously for every vector $v_j$ in $V$.

For any $z$ such that $\|z\|_2=1$, there exists a point $v_0$ (we change the subscript numbering for $V$ to index the order) in $V$ such
that $\|z-v_0\|_2\triangleq \gamma_1 \leq \gamma$. Let $z_1$ denote $z-v_0$,
 then $\|z_1-\gamma_1v_1\|_2 \triangleq \gamma_2 \leq \gamma_1 \gamma \leq \gamma^2$ for
 some $v_1$ in $V$. Repeating this process, we have $z=\sum_{j\geq 0} \gamma_j v_j$, where $\gamma_0=1$, $\gamma_j \leq \gamma^j$ and $v_j \in V$.

Thus for any $z \in R^m$, $z=\|z\|_2\sum_{j\geq 0} \gamma_j v_j$. For any index set $K$ with $|K| \leq \beta n$,
\begin{eqnarray*}
\sum \limits_{i \in K} |(Hz)_i| &=& \|z\|_2 \sum \limits_{i \in K} |(\sum \limits_{j \geq 0} \gamma_j Hv_j)_i| \\
& \leq & \|z\|_2 \sum \limits_{i \in K} \sum \limits_{j \geq 0} \gamma^{j} |(Hv_j)_i| \\
&= & \|z\|_2  \sum \limits_{j \geq 0} \gamma^{j} \sum \limits_{i \in K} |(Hv_j)_i| \\
& \leq & S\|z\|_2 \frac{(1-\delta)(1+\epsilon)}{(2-\delta){(1-\gamma)}}
\end{eqnarray*}

\begin{eqnarray*}
\sum \limits_{i} |(Hz)_i| &=& \|z\|_2 \sum \limits_{i } |(\sum \limits_{j \geq 0} \gamma_{j} Hv_j)_i| \\
& \geq & \|z\|_2 \sum \limits_{i}(|(Hv_0)_i|- \sum \limits_{j \geq 1} \gamma_{j} |(Hv_j)_i|) \\
& \geq & \|z\|_2 (\sum \limits_{i} |(Hv_0)_i|- \sum \limits_{j \geq 1} \gamma^{j} \sum \limits_{i}|(Hv_j)_i|) \\
& \geq & \|z\|_2 ((1-\epsilon)S- \sum \limits_{j \geq 1} \gamma^{j} (1+\epsilon)S) \\
& \geq &  S\|z\|_2 (1-\epsilon- \frac{\gamma(1+\epsilon)}{1-\gamma}).
\end{eqnarray*}

So $\sum \limits_ {i \in \overline{K}} |(Hz)_i| - \sum \limits_{i \in K}
|(Hz)_i| \geq S \|z\|_2
( 1-\epsilon- \frac{\gamma(1+\epsilon)}{1-\gamma}-2\frac{(1-\delta)(1+\epsilon)}{(2-\delta){(1-\gamma)}})$.
For a given $\delta$, we can pick $\gamma$ and $\epsilon$ small
enough such that $\sum \limits_ {i \in \overline{K}} |(Hz)_i| - \sum
\limits_{i \in K} |(Hz)_i| \geq \delta' S \|z\|_2$, satisfying the condition in Theorem \ref{thm:balanced}.
\end{proof}

If we do not require $\ell_1$ minimization to correct $k$ outliers over different supports, the fraction of outliers that are correctable can go to $1$.   
\begin{theorem}
 Take an arbitrary constant $0<\beta<1$ and let $\y=H\x+\e$, where $H$ is a Toeplitz matrix with Gaussian elements as defined earlier and $\e$ is \emph{a} vector with $k=\beta n$ nonzero elements. When $n \rightarrow \infty$, $\x$ can be recovered perfectly using $\ell_1$ minimization from $\e$ with $k\leq\beta n$ sparse errors with high probability.
\end{theorem}
%
%

%
\section{With Both Outliers and Observation Noises}
\label{sec:noiseerrortogether}
We further consider Toeplitz matrix based system identification when  both outliers and
random observation errors are present, namely, the observation $\y=H\x+\e+\w$,
where $\e$ is a sparse error with no more than $k$ nonzero elements and $\w$ is the vector of additive observation noises. We can show that error $\|\hat{\x}-\x\|_2$ goes to $0$ even when there are both outliers and random observation errors under mild conditions, where $\hat{\x}$ is the solution to (\ref{eq:errorcorrection}). 
\begin{theorem}
\label{thm:noiseerrortogether}
Let $m$ be a fixed positive integer and $H$ be an $n \times m$ Toeplitz matrix ($m<n$)
in (\ref{Top}) with each element $h_i$, $-m+2\leq i \leq n$, being i.i.d. $N(0,1)$ Gaussian random variables. Suppose $\y=H\x+\e+\w$, where $\e$ is a sparse
vector with $k \leq \beta n$ nonzero elements ($\beta<1$ is a constant) and $\w$ is
the observation noise vector. For any constant $t>0$, we assume that, with high probability as $n \rightarrow \infty$,  at least $\alpha(t)n$ (where $\alpha(t)>0$ is a constant depending on $t$ ) elements in $\w+\e$ are no bigger than $t$ in amplitude. Then $\|\hat{\x}-\x\|_2\rightarrow 0$ with high probability as $n \rightarrow \infty$, where $\hat{\x}$ is the solution to (\ref{eq:errorcorrection}).
\end{theorem}
\begin{proof}
$\|\y-H\hat{\x}\|_1$ can be written as $\|H(\x-\hat{\x})+\e+\w\|_1$. We argue that for any constant $t>0$, with high probability as $n \rightarrow 0$, for all $\hat{\x}$ such that $\|\x-\hat{\x}\|=t$, $\|H(\x-\hat{\x})+\e+\w\|_1 > \|\e+\w\|_1$, contradicting to $\hat{\x}$ being the solution to (\ref{eq:errorcorrection}).

To see this, we cover the sphere $Z=\{z| \|z\|_2=1\}$ with a $\gamma$-net $V$. We first argue that for every discrete point $tv_j$ with $v_j$ from the $\gamma$-net, $\|H tv_j+\e+\w\|_1 > \|\e+\w\|_1$; and then extend the result to the set $tZ$.

Let us denote $g(h,t)=\|H tv_j+\e+\w\|_1 - \|\e+\w\|_1=\sum_{i=1}^{n}(|l_i+t (Hv_j)_i|-|l_i|)$, where $l_i=(\e+\w)_i$ for $1\leq i \leq n$. We note that $(Hv_j)_i$ is a Gaussian random variable $N(0,1)$. Let $X$ be a Gaussian random variable $N(0,\sigma^2)$, then for an arbitrary $l$,
\begin{eqnarray*}
&&E\left\{ {|l+tX|}-|l| \right\}\\
&=&\frac{2}{\sqrt{2\pi}t\sigma}\int_{0}^{\infty} x e^{-\frac{(|l|+x)^2}{2t^2\sigma^2}}\,dx\\
&=& \sqrt{\frac{2}{\pi}}t\sigma e^{-\frac{l^2}{2t^2\sigma^2}}-2|l|(1-\Phi(\frac{|l|}{t\sigma})),
\end{eqnarray*}
which is a decreasing nonnegative function in $|l|$. From this, $E\{g(h,t)\}=\sum_{i=1}^{n}(\sqrt{\frac{2}{\pi}}t e^{-\frac{|l_i|^2}{2t^2}}-2|l_i|(1-\Phi(\frac{|l_i|}{t})))$. When $|l|\leq t$ and $\sigma=1$, $E\left\{ {|l+tX|}-|l| \right\}=\sqrt{\frac{2}{\pi}}t e^{-\frac{1}{2}}-2|l|(1-\Phi(1))\geq 0.1666t$.  It is also not hard to verify that $|g(a,t)-g(b,t)|\leq \sum_{i=1}^{n}t \sqrt{m}|a_i-b_i| \leq t\sqrt{mn}\|a_i-b_i\|_2$, and $g(h,t)$ has a Lipschitz constant (for $h$) no bigger than than $t\sqrt{mn}$.

Then by concentration of measure phenomenon for Gaussian random variables (see \cite{LedouxTalagrand1991, Ledoux01}),
\begin{eqnarray*}
&&P(g(h,t)\leq 0)\\
&=& P(\frac{g(h,t)-E\{g(h,t)\}}{t\sqrt{mn}} \leq -\frac{E\{g(h,t)\}}{t\sqrt{mn}} )\\
&\leq& 2e^{-\frac{ \left(\sum_{i=1}^{n}\left[\sqrt{\frac{2}{\pi}}t e^{-\frac{l_i^2}{2t^2}}-2|l_i|(1-\Phi(\frac{|l_i|}{t}))\right]\right)^2  }{2t^2 nm}}\triangleq2e^{-B}.
\end{eqnarray*}
If there exists a constant $\alpha (t)$ such that, as $n \rightarrow \infty$, at least $\alpha(t) n$ elements have magnitudes smaller than $t$, then the numerator in $B$ behaves as $\Theta(n^2)$ and the corresponding probability $P(g(h,t)\leq 0)$ behaves as $2e^{-\Theta( n)}$. This is because when $|l|\leq t$, $\sqrt{\frac{2}{\pi}}t e^{-\frac{|l|^2}{2t^2}}-2|l|(1-\Phi(\frac{|l|}{t})) \geq 0.1666 t$.

By the same reasoning, $g(h,t)\leq \epsilon n$ holds with probability no more than $e^{-c_5 n}$ for each discrete point from the $\gamma$-net $tV$, where $\epsilon>0$ is a sufficiently small constant and $c_5>0$ is a constant which may depend on $\epsilon$.  Since there are at most $(1+\frac{2}{\gamma})^{m}$ points from the $\gamma$-net, by a simple union bound, with probability $1-e^{-c_6 n}$ as $n \rightarrow \infty$, $g(h,t)> \epsilon n$ holds for \emph{all} points from the $\gamma$-net $tV$, where $c_6>0$ is a constant and $\gamma$ can be taken as an arbitrarily small \emph{constant}. Following similar $\gamma$-net proof techniques for Lemmas \ref{cor:s1}, \ref{lemma:blconcentration} and \ref{lemma:slp}, if we choose a sufficiently small constant $\epsilon>0$ and accordingly a sufficiently small constant $\gamma>0$, $g(h,t)> 0.5 \epsilon n$ holds simultaneously for every point in the set $tZ$ with high probability $1-e^{-c_7 n}$, where $c_7>0$ is a constant.

Notice if $g(h, t)>0$ for $t=t_1$, then necessarily $g(h,t)>0$ for $t=t_2>t_1$. This is because $g(h,t)$ is a convex function in $t\geq 0$ and $g(h,0)=0$. So if $g(h,t)> 0.5 \epsilon n>0$ holds with high probability for every point $tZ$, necessarily $\|\hat{\x}-\x\|_2<t$, because $\hat{\x}$ minimizes the objective in (\ref{eq:errorcorrection}). Because we can pick $t$ to be arbitrarily small, $\|\hat{\x}-\x\|_2\rightarrow 0$ with high probability as $n \rightarrow \infty$.
\end{proof}

We remark that the mild conditions in Theorem \ref{thm:noiseerrortogether} are satisfied easily if $\beta<1$ and the elements in $\w$ are i.i.d. random variables following a probability density function $f(s)$ that is not ``vanishing'' around $s=0$ (namely the cumulative distribution function $F(t)>0$ for any $t>0$. $f(0)$ can be $0$ sometimes). For example, Gaussian distribution, exponential distributions, and Gamma distributions for $\w$ all satisfy such conditions in Theorem \ref{thm:noiseerrortogether}. This greatly broadens the results in \cite{SIC}, which requires $f(0)>0$ and does not accommodate outliers. Compared with analysis in compressed sensing \cite{CandesErrorCorrection, DonohoNoise}, this result is for Toeplitz matrix in error correction and applies to observation noises with non-Gaussian distributions.

\section{Numerical Evaluations}
\label{sec:numerical}
Based on Theorem \ref{thm:strongmain}, we calculate the strong thresholds in Figure \ref{fig:gaussianfixrho} for different values of $m$ by optimizing over $\mu>0$ and $\delta$. As $m$ increases, the correlation length in the matrix $H$ also increases and the corresponding correctable number of errors decreases (but always exists).
\begin{figure}[t]
\centering
\includegraphics[width=2.5in, height=1.2in]{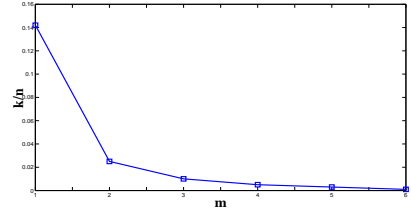}
\caption{Recoverable fraction of errors versus $m$}\label{fig:gaussianfixrho}
\end{figure}
\begin{figure}[t]
\centering
\includegraphics[width=2.5in, height=1.5in]{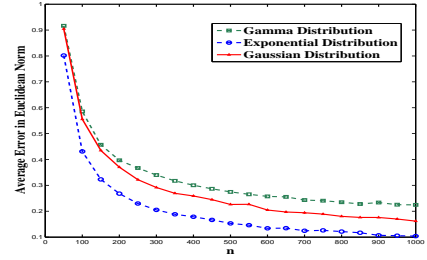}
\caption{With outliers and noises of different distributions}
\label{fig:together}
\end{figure}
We then evaluate in Figure \ref{fig:together} the $\ell_2$-norm error $\|\hat{\x}-\x\|_2$ of $\ell_1$ minimization for Gaussian Toeplitz matrices under both outliers and i.i.d. observation noises of different probability distributions: Gamma distribution with shape parameter $k=2$ and scale $\frac{1}{\sqrt{6}}$; standard Gaussian distribution $N(0,1)$ and exponential distribution with mean $\frac{\sqrt{2}}{2}$. These distributions are chosen such that the observation noises have the same expected energy. The system parameter $m$ is set to $5$ and the system state $\x$ are generated as i.i.d. standard Gaussian random variables. We randomly pick $\frac{n}{2}$ i.i.d. $N(0,100)$ Gaussian outliers with random support for the error vector $\e$. For all these distributions, the average error goes to $0$ (we also verified points beyond $n>1000$). What is interesting is that the error goes to $0$ at different rates. Actually, as hinted by the proof of Theorem \ref{thm:noiseerrortogether}, the Gamma distribution has the worst performance because its probability density function is smaller around the origin (actually $0$ at the origin), while the exponential distribution has the largest probability density function around $0$.

\bibliographystyle{IEEEbib}

\end{document}